%% file: main.tex
\documentclass[letterpaper, 10pt, conference]{ieeeconf}
\IEEEoverridecommandlockouts
\usepackage{amsthm}
\usepackage{amsmath, amssymb}
\usepackage{amsfonts}
\usepackage{mathtools} \mathtoolsset{showonlyrefs}
\usepackage{graphicx}
\usepackage{hyperref}
\usepackage{empheq}
\usepackage{siunitx}
\usepackage{pgfplots}\pgfplotsset{compat=1.18}

\usepackage{graphicx}
\newcommand{\showfigscale}[2][]{%
  \begingroup
  \sbox0{\includegraphics[#1]{#2}}%
  \typeout{FIG '#2': target width=\the\wd0, target height=\the\ht0}%
  \sbox1{\includegraphics{#2}}%
  \typeout{FIG '#2': natural width=\the\wd1, natural height=\the\ht1}%
  \typeout{FIG '#2': linewidth=\the\linewidth}%
  \endgroup
}

\newcommand{\norm}[1]{\lVert #1 \rVert}
\newcommand{\R}{\mathbb{R}}
\newcommand{\N}{\mathbb{N}}
\theoremstyle{plain}
\newtheorem{theorem}{Theorem}

\newtheorem{lemma}{Lemma}
\newtheorem{assumption}{Assumption}
\newtheorem{corollary}{Corollary}

\DeclareMathOperator{\diag}{diag}
\DeclareMathOperator{\atantwo}{atan2}

\theoremstyle{definition}
\newtheorem{definition}{Definition}
\theoremstyle{remark}
\newtheorem*{remark}{Remark}

\title{\LARGE \bf Over-the-Air Consensus-based Formation Control of Heterogeneous Agents: Communication-Rate and Geometry-Aware Convergence Guarantees}

\author{Michael Epp, Fabio Molinari, Jörg Raisch%
\thanks{Michael Epp, Fabio Molinari and Jörg Raisch are with the Control Systems Group, Technische Universität Berlin, Germany. Email: {\tt\small michael.epp@tu-berlin.de}}%
\thanks{Fabio Molinari is also with Nimbus AI Technologies Inc., Seattle, USA.}
\thanks{Jörg Raisch is also with Science of Intelligence (SCIoI), Research Cluster of Excellence, Berlin, Germany. Email: {\tt\small raisch@control.tu-berlin.de}}
\thanks{This work was funded by the Federal Ministry of Education and Research of Germany joint project 6G-RIC, project identification number 16KISK030.}
}

\begin{document}
\maketitle
\begin{abstract}
This paper investigates the formation control problem of heterogeneous, autonomous agents that communicate over a wireless multiple access channel.
Instead of avoiding interference through orthogonal node-to-node transmissions, we exploit the superposition property of the wireless channel to compute, at each receiver, normalized convex combinations of simultaneously broadcast neighbor signals.
At every communication instant, agents update their reference positions from these aggregates, and track the references in continuous time between updates.
The only assumption on the agent dynamics is that each agent tracks constant reference positions exponentially, which accommodates a broad class of platforms.
Under this assumption, we analyze the resulting jump-flow system under time-varying communication graphs and unknown channel coefficients.
We derive a communication-rate based sufficient condition that guarantees convergence to a prescribed formation.
We then provide a geometry-aware refinement showing how favorable tracking transients can relax the required condition.
Simulations with unicycle agents illustrate the theoretical results and demonstrate a substantial reduction in the number of required orthogonal transmissions compared to interference-avoiding node-to-node communication protocols.
\end{abstract}

\input{src/01_introduction}
\input{src/02_preliminaries}
\input{src/03_communication}
\input{src/04_control}
\input{src/05_analysis}
\input{src/06_simulation}
\input{src/07_conclusion}

\bibliographystyle{ieeetr}
\bibliography{lib}
\end{document}

%% file: src/01_introduction.tex
\section{Introduction} \label{sec:introduction}
Autonomous multi-agent systems (MAS) have gained significant attention in science and engineering, motivated by applications such as cooperative robotics, unmanned vehicle teams, and networked sensing.
A key driver is that many collective objectives can be achieved without centralized coordination by using distributed control laws that rely only on local information exchange, while remaining scalable and robust to failure and topology changes~\cite{ren2007, olfati-saber2007, gulzar2018}.

Formation control is a central coordination problem in MAS, where the objective is to achieve and maintain a prescribed spatial arrangement using only local interaction rules.
Two broad paradigms are commonly considered.
Inter-agent measurements can be used to encode the desired formation in relative measurements between neighboring agents, without a global reference variable~\cite{dimarogonas2008,falconi2015}.
On the other hand, in centroid-based formation control, agents agree on a common reference point and realize the formation through prescribed displacement vectors with respect to that reference~\cite{ren2007b, xie2014}.
In the latter case, distributed consensus mechanisms are a natural tool to establish agreement on the shared reference point.
Consensus algorithms under directed and time-varying interaction topologies have been extensively studied, providing conditions that connect connectivity properties of the communication graph to convergence guarantees~\cite{olfati-saber2007, chatterjee1977, tahbaz-salehi2008}.

In many robotic systems, wireless communication resources are limited.
Traditional multi-agent coordination controls often presume that pairwise communication can be implemented reliably, implicitly relying on interference-avoidance mechanisms.
In dense swarms, however, allocating orthogonal channels to all links is costly and can introduce latency that degrades closed-loop performance.
An alternative is leveraging the superposition property of wireless signals, a concept known as over-the-air (OTA) computation~\cite{goldenbaum2013}.
This approach, identified as a promising candidate technology for 6G wireless communication systems~\cite{wang2024}, turns traditional interference challenges into a communication advantage.

For MAS, OTA computation naturally enables communication-efficient consensus control.
In particular, \cite{molinari2019} proposed a consensus protocol that uses simultaneous broadcasts and receiver-side normalization so that each agent can obtain a convex combination of its neighbors' values without decoding individual transmissions, reducing the need of orthogonal channels and offering robustness to unknown channel gains.
Building on this principle, consensus-based formation control for single-integrator agents was studied in~\cite{molinari2019} and later extended to include collision avoidance in~\cite{epp2024}.
The consensus-based formation control of non-holonomic agents was studied in~\cite{borzone2018}, which established sufficient conditions under which a swarm of agents establishes the formation.

This paper develops an OTA, consensus-based formation control scheme for heterogeneous agents under sampled wireless communication.
Rather than assuming identical agent dynamics, we require only that each agent's position is exponentially stabilizable to any constant reference, covering a broad range of platforms, including unicycle and differential drive models~\cite{dewit1992} and car-like models~\cite{sordalen1995}.
For the resulting jump-flow system, we derive a communication-rate based sufficient condition for convergence to the desired formation and provide a geometry-aware refinement that captures how the tracking transient can relax the required communication-rate bound, thereby generalizing earlier formation control results in~\cite{molinari2019} and~\cite{borzone2018}.

The remainder of this paper is structured as follows.
Notation is summarized at the end of this section.
Section~\ref{sec:preliminaries} introduces the problem formulation and basic assumptions. 
In Section~\ref{sec:communication}, the OTA broadcast protocol and its properties are described.
A control strategy for heterogeneous agents is proposed in Section~\ref{sec:control}.
The closed-loop system is analyzed and convergence conditions are established in Section~\ref{sec:analysis}.
Simulation results are shown in Section~\ref{sec:simulation}, and Section~\ref{sec:conclusion} concludes the paper.

\subsection*{Notation}
Throughout this paper, $\R$, $\R_{>0}$, and $\R_{\geq 0}$ denote the sets of real, positive real and nonnegative real numbers, respectively. 
$\N_0$ and $\N$ will denote the nonnegative and positive integers, respectively. 
$I_n$ is the identity matrix of size $n\times n$, while $1_n$ is a column vector of size $n$ with every element equal to $1$.
We denote the Euclidean norm by $\norm{.}$.

Given a matrix $A$, its transpose is written as $A^\top$.
The entry in position $(i,j)$ of matrix $A$ is denoted $[A]_{ij}$. 
A matrix $A$ is positive if $\forall (i,j),\, [A]_{ij}>0$, and nonnegative if $[A]_{ij}\geq 0$.
A square matrix $A\in\R^{n\times n}$ is called \emph{row-stochastic} if it is nonnegative and satisfies $A1_n = 1_n$.
The coefficient of ergodicity of a row-stochastic matrix $A$ is defined as $\tau_1(A)=\frac{1}{2}\max_{i,j}\sum_{s=1}^n |A_{is}-A_{js}|$.

A \emph{directed graph} $\mathcal{G}$ is a pair $(\mathcal{N}, \mathcal{A})$, where $\mathcal{N}$ is the set of nodes and $\mathcal{A}\subseteq \mathcal{N}\times\mathcal{N}$ is the set of arcs, i.e., $(i,j)\in\mathcal{A}$ if and only if an arc goes from node $i\in\mathcal{N}$ to node $j\in\mathcal{N}$. 
Node $j$ is a neighbor of node $i$ if $(j,i)\in\mathcal{A}$, and $\mathcal{N}_i = \left\{ j\in\mathcal{N} \,\vert\, (j,i)\in\mathcal{A}  \right\}$ denotes the set of neighbors of node $i$.
A path from node $i$ to node $j$ is a sequence of arcs 
\begin{equation}
    (l_0,l_1),(l_1,l_2),\dots,(l_{p-1},l_p),
\end{equation}
with $p\in\N$, $l_0=i$ and $l_p=j$. 
The graph $\mathcal{G}$ is strongly connected if $\forall i,j\in\mathcal{N}$, there exists a path from node $i$ to node $j$.
A weighted directed graph is a triple $\mathcal{G} = (\mathcal{N}, \mathcal{A}, w)$, where $w:\mathcal{A}\to\R_{>0}$ assigns a positive weight to each arc.
Let $n=|\mathcal{N}|$, the matrix $\mathcal{W}\in\R_{\geq 0}^{n\times n}$, with $[\mathcal{W}]_{ji}=w\left((i,j)\right)$ if $(i,j)\in\mathcal{A}$ and $[\mathcal{W}]_{ji}=0$ otherwise, is called the adjacency matrix of $\mathcal{G}$.
In this paper, we use weighted directed graphs to model the communication topology in multi-agent systems.
We will identify each agent with a node from $\mathcal{G}$, and $(i,j)\in\mathcal{A}$ indicates that agent $i$ can transmit information to agent $j$.
Allowing for time-varying communication topologies, we write $\mathcal{G}(t) = (\mathcal{N}, \mathcal{A}(t), w(t))$ and $\mathcal{N}_i(t)=\{j\in\N, \mid (j,i)\in\mathcal{A}(t)\}$.
We consider communication instants $\{t_k\}, k\in\N_0$.
For any time-varying quantity $x(t)$, we denote its sampled value at $t_k$ by $x_k \coloneqq x(t_k)$.

%% file: src/02_preliminaries.tex
\section{Preliminaries} \label{sec:preliminaries}
Let $\mathcal{N}=\{1,\dots,n\}$ index a set of heterogeneous autonomous agents moving on a two-dimensional plane.
Agent $i$ has position $p_i(t)\in\R^2$ and control input $u_i(t)\in\R^{m_i}$, with $m_i\in\N$, and evolves according to
\begin{align} 
\dot{x}_i(t) &= f_i \left(t, x_i(t), u_i(t) \right) , \quad x_i(t)\in\R^{n_i} \label{eq:agent_dynamics_in}\\
p_i(t) &= g_i(t,x_i(t)), \label{eq:agent_dynamics_out}
\end{align}
where $f_i$ and $g_i$ are functions of appropriate dimensions.
This formulation allows heterogeneous agent dynamics, while the formation objective described in terms of the position output $p_i(t)$.
Moreover, we assume the agents are dynamically independent in the sense that~\eqref{eq:agent_dynamics_in} and \eqref{eq:agent_dynamics_out} do not depend on other agents' states.
Inter-agent coupling is solely introduced through the feedback law used to generate $u_i(t)$.
Additionally, we impose the following assumption on the controller.
\begin{assumption} \label{as:expo_stable}
    For each agent $i\in\mathcal{N}$ there exist constants $C_i,\lambda_i\in\R_{>0}$ such that, for any constant reference position $r_i\in\R^2$, there exists an input signal $u_i(t)$ for which the corresponding solution satisfies, for all $t\geq\tau\geq 0$,
    \begin{equation}
        \| p_i(t) - r_i \| \leq C_i e^{-\lambda_i (t-\tau)} \| p_i(\tau)-r_i \|.
    \end{equation}
\end{assumption}
\begin{remark}
    At this point, we only impose exponential stability of an agent's position to its reference. 
    This assumption is satisfied by many common mobile robot models such as differential drive and car-like kinematics under regulation control~\cite{dewit1992, sordalen1995}.
    
    However, we do not impose any additional geometric properties on the resulting trajectories.
    In Section~\ref{sec:analysis}, we derive convergence conditions that admit a geometric interpretation, indicating how properties of the tracking behavior affect consensus.
\end{remark}

The agents exchange information at sampling times $t_k\in\R_{\geq 0}, k\in\N_0$, with uniformly bounded inter-sampling times, i.e., there exist $\check T, \hat T \in\R_{>0}$, such that
\begin{equation}
\forall k\in\N_0,~ \check{T} \leq t_{k+1}-t_k \leq \hat{T}.
\end{equation}

We aim to achieve a prescribed formation described by displacement vectors $d_i\in\R^2$.
This is achieved asymptotically if, 
\begin{equation} \label{eq:form_cond_convergence}
    \forall i\in\mathcal{N}, \quad \lim_{t\to\infty} p_i(t) = c(t)+d_i,
\end{equation}
where $c(t)\in\R^2$ is the so-called centroid of the formation.

%% file: src/03_communication.tex
\section{Communication Model} \label{sec:communication}
This study adopts the WMAC model to describe the superposition property (interference) of the wireless medium, i.e., the effect of multiple signals being simultaneously transmitted in the same frequency band. 
Interference has traditionally been avoided by using more wireless resources, e.g., by creating orthogonal transmission (each agent is assigned its own time or frequency slot). 
Instead of avoiding interference, however, we let agents broadcast simultaneously, giving rise to the following definition.
\begin{definition}[Wireless Multiple Access Channel] \label{def:wmac}
    Let all agents in  $\mathcal{N}$ simultaneously broadcast $\alpha_j(t_k)\in\mathbb{R}$. 
    The obtained superimposed value at the receiving agent is
    \begin{equation} \label{eq:wmac_def}
        y_i(t_k) = \sum_{j\in\mathcal{N}_i(t_k)} \xi_{ij}(t_k)\alpha_j(t_k),
    \end{equation}
    where $\xi_{ij}(t_k)\in (0,1]$ are unknown time-varying channel coefficients.
    Since the received superimposed value is a result of power modulation, the channel coefficients are assumed to be positive.
\end{definition}

For normalization purposes,~\cite{molinari2018} proposed to additionally broadcast a known value, e.g., $1$ via an orthogonal channel (i.e., in a different time slot or frequency range).
To transmit a vector $\mu_j\in\mathbb{R}^m$, we apply the WMAC model element-wise and allocate an orthogonal channel to each component.
If we use TDMA (Time Division Multiple Access), we assume that the delays between the individual broadcasts of the entries of the vector $\mu_j$ and the normalization signal are so small that they can be considered to occur at the same time $t_k$, i.e., the channel coefficients will not change between these broadcasts.

According to the WMAC model and using the notation $\nu_i(t_k)=\nu_{i,k}$, all agents $i\in\mathcal{N}$ then receive for all $k\in\N_0$ 
\begin{align}
    \nu_{i,k} &= \sum_{j\in\mathcal{N}_{i,k}} \xi_{ij,k}\mu_{j,k} \\
    \nu_{i,k}' &= \sum_{j\in\mathcal{N}_{i,k}} \xi_{ij,k}.
\end{align}
In the following, we will assume that $\forall k\in\N_0$, $\forall i\in\mathcal{N}, i\in\mathcal{N}_i$, i.e., $\xi_{ii,k}\in\R_{>0}$.
Therefore $\nu_{i,k}'$ is positive and the vector
\begin{equation}
    \zeta_{i,k} \coloneqq \frac{\nu_{i,k}}{\nu_{i,k}'}
\end{equation}
is well defined. 
Clearly,
\begin{equation}
    \zeta_{i,k} = \sum_{j\in\mathcal{N}_{i,k}} h_{ij,k} \mu_{j,k},
\end{equation}
where
\begin{equation}
    h_{ij,k} = \begin{cases}
        \frac{\xi_{ij,k}}{\sum_{j\in\mathcal{N}_{i,k}} \xi_{ij,k}} & \text{if } (j,i)\in\mathcal{A} \\
        0 & \text{otherwise}
    \end{cases}
    \label{eq:h_def}
\end{equation}
are the normalized channel coefficients.

By construction, $h_{ij,k}\in[0,1]$, and
\begin{equation}
    \sum_{j=1}^n h_{ij,k} = 1. \label{eq:h_sum}
\end{equation}
Hence, $\zeta_{i,k}$ is in the convex hull of the $\mu_{j,k}$, $j\in\mathcal{N}$.
It encapsulates aggregated information from neighboring agents.

\begin{corollary}
    For all $k\in\N_0$, the matrix $H_k$, with $[H_k]_{ij}=h_{ij,k}$, is row-stochastic.
\end{corollary}
\begin{assumption} \label{as:mixing}
    There exist $L\in\N$ and $\mu\in\R>0$ such that for all $k\in\N_0$
    \begin{equation}
        \tau_1(H_{k+L-1}\cdots H_k)\leq 1-\mu.
    \end{equation}
\end{assumption}
\begin{remark}
    Assumption~\ref{as:mixing} is a standard connectivity requirement for consensus over time-varying networks.
    Within every block of $L$ consecutive updates, the product $H_{k+L-1}\cdots H_k$ is sufficiently ergodic to ensure a uniform contraction of disagreement.
    The assumption is satisfied, for example, under joint strong connectivity of the communication graphs together with a uniform positive lower bound on the nonzero weights of $H_k$.
    These conditions are classical in distributed control and exclude switching patterns that keep the network essentially disconnected over arbitrarily long periods.
\end{remark}

%% file: src/04_control.tex
\section{Control Strategy} \label{sec:control}
In this section we present the reference-update rule induced by the communication model described in Section~\ref{sec:communication}.
For each communication instant $t_k$, let $p_{i,k^+}$ and $r_{i,k^+}$ denote, respectively, the position of agent $i$ and its reference immediately after the update.
The resulting closed-loop system consists of jump and flow dynamics, accounting for the discrete-time communication instants and the continuous-time motion in between them.

\subsection*{Jump Dynamics}
At each communication instant $t_k$, agents broadcast their displacement-shifted position $p_{i,k}-d_i$.
The reference position for agent $i$ is then updated as, for all $k\in\N_0$,
\begin{equation} \label{eq:jump_r}
    r_{i,k^+} = d_i + \sum_{j=1}^n h_{ij,k} (p_{j,k}-d_j).
\end{equation}
Since the physical position does not jump at communication instants, we have, for all $k\in\N_0$, 
\begin{equation} \label{eq:jump_p}
    p_{i,k^+}=p_{i,k}.
\end{equation}

\subsection*{Flow Dynamics}
Between two adjacent communication instants, the references are held constant, i.e., for all $t\in(t_k,t_{k+1}], k\in\N_0$,
\begin{equation} \label{eq:flow_r}
    r_{i}(t) = r_{i,k^+},
\end{equation}
and in particular $r_{i,k+1}=r_{i,k^+}$.
Under Assumption~\ref{as:expo_stable}, the agents then converge exponentially to their individual reference positions, i.e., for all $k\in\N_0$,
\begin{equation} \label{eq:flow_p}
    \| p_{i,k+1}-r_{i,k^+} \| \leq C_i e^{-\lambda_i (t_{k+1}-t_k)} \| p_{i,k}-r_{i,k^+} \|.
\end{equation}

%% file: src/05_analysis.tex
\section{Convergence Analysis} \label{sec:analysis}
In this section we first provide a general bound on $\check T$ under which the system achieves convergence and establishes the desired formation.
We then continue the analysis by exploiting additional geometric information about the tracking transient.

\subsection*{General Convergence Condition}
We start by introducing the displacement-shifted state variables $\tilde p_i(t) \coloneqq p_i(t)-d_i$ and $\tilde r_{i,k} \coloneqq r_{i,k}-d_i$.
The jump-flow system in~\eqref{eq:jump_r} to~\eqref{eq:flow_p} can then be rewritten as
\begin{align}
    \tilde r_{i,k^+} &= \sum_{j=1}^n h_{ij,k} \tilde p_{j,k}, \label{eq:jump_r_tilde}\\
    \tilde p_{i,k^+} &= \tilde p_{i,k}
\end{align}
at jumps for all $k\in\N_0$ and
\begin{align}
    \tilde r_i(t) &= \tilde r_{i,k^+}, \label{eq:flow_r_tilde} \\
    \| \tilde p_{i,k+1}-\tilde r_{i,k^+} \| &\leq C_i e^{-\lambda_i(t_{k+1}-t_k)} \| \tilde p_{i,k} - \tilde r_{i,k^+} \| \label{eq:flow_p_tilde}
\end{align}
during flow intervals for all $t\in (t_k,t_{k+1}]$.
By~\eqref{eq:form_cond_convergence}, the desired formation is therefore achieved if and only if $\forall i,j,\, \tilde p_i(t)=\tilde p_j(t)$.

Additionally, we denote stacked vectors by omitting the agent related subscript, i.e., we have $\tilde p_k = [\tilde p_{1,k}^\top,\dots,\tilde p_{n,k}^\top]^\top$ and $\tilde r_k = [\tilde r_{1,k}^\top,\dots,\tilde r_{n,k}^\top]^\top$.

In order to analyze the disagreement in the system, we define the seminorm
\begin{equation} \label{eq:Delta_def}
    \Delta(x) = \max_{i,j} \| x_i-x_j \|,
\end{equation}
for any stacked vector $x = [x_1^\top,\dots,x_n^\top]^\top$, $\forall i,\, x_i\in\R^2$.
It satisfies the usual seminorm properties
\begin{align}
    \Delta(x) &\geq 0, \\
    \Delta(\alpha x) &= |\alpha|\Delta(x),\,\alpha\in\R, \\
    \Delta(x + y) &\leq \Delta(x) + \Delta(y),
\end{align}
and in addition 
\begin{equation}
    \Delta(x)=0 \iff \forall i,j,\, x_i = x_j.
\end{equation}
The seminorm~\eqref{eq:Delta_def} can therefore be used as a measure of disagreement in a system and is zero exactly when that system reaches consensus.
Additionally, for any row-stochastic matrix $A\in\R^{n\times n}$, we have by~\cite[Theorem 3.1]{seneta2006},
\begin{equation}
    \Delta((A\otimes I_2)x) \leq \tau_1(A)\Delta(x).
\end{equation}

With these definitions at hand, we are able to state the following theorem.
\begin{theorem} \label{thm:general}
    Suppose Assumptions~\ref{as:expo_stable} and~\ref{as:mixing} are satisfied.
    Then there exists a constant $\check T^* \in\R_{>0}$, such that if $\check T > \check T^*$, i.e., 
    \begin{equation}
        \forall k,\, t_{k+1}-t_k > \check T^*,
    \end{equation}
    the jump-flow system in~\eqref{eq:jump_r} to~\eqref{eq:flow_p} achieves the desired formation in the sense of~\eqref{eq:form_cond_convergence}.
\end{theorem}
\begin{proof}
    By substituting~\eqref{eq:jump_r_tilde} in~\eqref{eq:flow_p_tilde}, we obtain, for all $i\in\mathcal{N}$,
    \begin{align}
        \| \tilde p_{i,k+1}-\tilde r_{i,k^+} \| &\leq C_i e^{-\lambda_i(t_{k+1}-t_k)} \bigg\| \tilde p_{i,k} - \sum_{j=1}^n h_{ij,k} \tilde p_{j,k} \bigg\| \nonumber\\
        &\leq C_i e^{-\lambda_i \check T} \Delta(\tilde p_k),
    \end{align}
    where we used that
    \begin{align}
        \bigg\| \tilde p_{i,k} - \sum_{j=1}^n h_{ij,k} \tilde p_{j,k} \bigg\| &\leq \max_l \bigg\| \tilde p_{i,k} - \sum_{j=1}^n h_{ij,k} \tilde p_{l,k} \bigg\| \nonumber\\
        &= \max_l \| \tilde p_{i,k} - \tilde p_{l,k} \| 
        \leq \Delta(\tilde p_k)
        \label{eq:p_r_Delta_ineq}
    \end{align}
    and $e^{-\lambda_i(t_{k+1}-t_k)}\leq e^{-\lambda_i \check T}$.
    We now define 
    \begin{equation} \label{eq:beta_i_def}
        \beta_i = C_i e^{-\lambda_i \check T}
    \end{equation}
    and $\hat\beta = \max_i \beta_i$ to obtain, for all $i\in\mathcal{N}$,
    \begin{equation}
        \| \tilde p_{i,k+1}-\tilde r_{i,k^+} \| \leq \hat\beta \Delta(\tilde p_k).
    \end{equation}
    By defining $w_k = \tilde p_{k+1} - \tilde r_{k^+}$ and making use of the subadditivity property of $\|\cdot\|$, we hence obtain
    \begin{align}
        \Delta(w_k) &= \max_{i,j} \| (\tilde p_{i,k+1}-\tilde r_{i,k^+}) - (\tilde p_{j,k+1}-\tilde r_{j,k^+}) \| \nonumber\\
        &\leq 2\max_{i} \| \tilde p_{i,k+1}-\tilde r_{i,k^+} \| \nonumber\\
        &\leq 2\hat\beta \Delta(\tilde p_k).
        \label{eq:Delta_w}
    \end{align}

    Using~\eqref{eq:jump_r_tilde}, we can write $\tilde p_{k+1} = (H_k\otimes I_2)\tilde p_k + w_k$.
    By iterating this recursion over an $L$-step window, applying $\Delta(\cdot)$, and making use of its properties, we obtain
    \begin{align}
        \Delta(\tilde p_{k+L}) &\leq \tau_1\bigg(\prod_{\ell=1}^L H_{k+L-\ell}\bigg)\Delta(\tilde p_k) \nonumber \\
        &\quad + \sum_{i=1}^L \tau_1\bigg(\prod_{\ell=1}^{i-1} H_{k+L-\ell} \bigg)\Delta(w_{k+L-i}) \nonumber\\
        &\leq (1-\mu)\Delta(\tilde p_k) + \sum_{i=1}^L \Delta(w_{k+L-i}) \nonumber\\
        &\leq (1-\mu)\Delta(\tilde p_k) + 2\hat\beta \sum_{i=1}^L \Delta(\tilde p_{k+L-i}) \nonumber\\
        &\leq (1-\mu)\Delta(\tilde p_k) + 2\hat\beta \sum_{i=1}^L (1+2\hat\beta)^{i-1} \Delta(\tilde p_{k}) \nonumber\\
        &= (1-\mu)\Delta(\tilde p_k) + \big( (1+2\hat\beta)^L -1 \big) \Delta(\tilde p_{k}),
    \end{align}
    where in the first inequality we used the convention $\prod_{\ell=1}^{i-1}H_{k+L-\ell}=H_{k+L-1}\cdots H_{k+L-i+1}$ and for $i=0$ this product is empty and therefore equal to the identity matrix. 
    In the third inequality we used the one-step relation
    \begin{align}
        \Delta(\tilde p_{k+1}) &\leq \tau_1(H_k)\Delta(\tilde p_k) + \Delta(w_k) \nonumber\\
        &\leq (1+2\hat\beta) \Delta(\tilde p_k).
    \end{align}

    In order to establish convergence to the formation, it therefore suffices to require
    \begin{equation}
        (1-\mu) + \big( (1+2\hat\beta)^L -1 \big) < 1
    \end{equation}
    or equivalently
    \begin{equation}
        \hat\beta < \frac{1}{2}\big((1+\mu)^{\tfrac{1}{L}}-1\big).
    \end{equation}
    Since $\hat\beta \leq \hat{C} e^{-\check\lambda \check T}$, where $\hat{C}=\max_i C_i$ and $\check\lambda=\min_i \lambda_i$, a sufficient condition for convergence is
    \begin{equation} \label{eq:conv_cond_general}
        \check T>\check T^* \coloneqq -\frac{1}{\check\lambda}\ln\frac{(1+\mu)^{\tfrac{1}{L}}-1}{2 \hat{C}}.
    \end{equation}
    This concludes the proof of Theorem~\ref{thm:general}.
\end{proof}

The jump-flow system in~\eqref{eq:jump_r} to~\eqref{eq:flow_p} therefore converges if the time in between communication instants is sufficiently large.
However, the requirement in~\eqref{eq:conv_cond_general} can be conservative in practice since increasing the inter-communication time may be undesirable, as in many settings the sampling schedule is fixed, and increasing $\check T$ might lead to reduced responsiveness of the system.
In the following, we therefore show how the geometric features of the tracking transient affect the convergence.

\subsection*{Geometry-Aware Convergence Condition}
We start by relaxing Assumption~\ref{as:expo_stable}.
\begin{assumption} \label{as:expo_stable_relaxed}
    For each agent $i\in\mathcal{N}$ there exist constants $C_i,\lambda_i\in\R_{>0}$ and scalar $\sigma_i\in(0,1]$ such that the following holds.
    For any constant reference position $r_i\in\R^2$ and any initial time $\tau\in\R_{\geq 0}$, define the point
    \begin{equation}
        s_i(\tau) = (1-\sigma_i)p_i(\tau) + \sigma_i r_i.
    \end{equation}
    Then there exists an input signal $u_i(t)$ for which the corresponding solution satisfies, for all $t\geq\tau$,
    \begin{align} \label{eq:seg_expo_bound}
        \| p_i(t)-s_i(\tau) \| &\leq C_i e^{-\lambda_i(t-\tau)} \| p_i(\tau)-s_i(\tau) \|.
    \end{align}
\end{assumption}
This assumption relaxes Assumption~\ref{as:expo_stable} by requiring exponential convergence only to a point $s_i(\tau)$ on the line segment joining the initial position $p_i(\tau)$ and the reference $r_i$, rather than to the reference itself.
The parameter $\sigma_i$ determines how far that point lies toward $r_i$ and the case $\sigma_i=1$ recovers Assumption~\ref{as:expo_stable}.
This is illustrated in Figure~\ref{fig:set_comp}.
The dashed circle depicts the guaranteed end-of-interval region under Assumption~\ref{as:expo_stable}.
The entire blue shaded area shows the enlarged guaranteed end-of-interval region under Assumption~\ref{as:expo_stable_relaxed}.
\begin{figure}
    \centering
    \includegraphics[width=\linewidth]{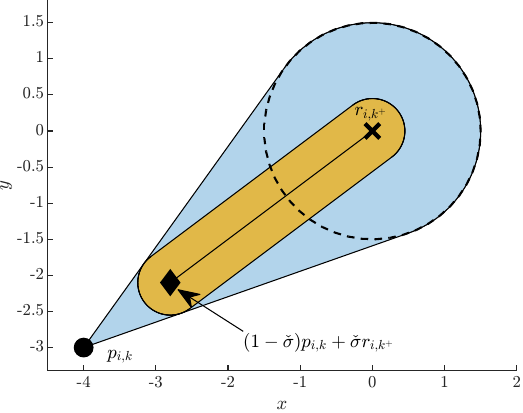}
    \caption{Comparison of guaranteed end-of-interval regions.
    Under Assumption~\ref{as:expo_stable}, this is for agent $i$, a ball centered at $r_{i,k^+}$ shown as the dashed circle.
    Assumption~\ref{as:expo_stable_relaxed} relaxes the guaranteed end-of interval region to the entire blue shaded area.
    The yellow set is the guaranteed end-of-interval region under which Theorem~\ref{thm:relaxed} formulates a sufficient consensus condition.}
    \label{fig:set_comp}
\end{figure}

We define the segment point
\begin{equation} \label{eq:s_i_def}
    s_{i,k^+} = (1-\sigma_{i,k})p_{i,k} + \sigma_{i,k} r_{i,k^+}.
\end{equation}
Evaluating~\eqref{eq:seg_expo_bound} at $t=t_{k+1}$ and using $t_{k+1}-t_k \geq \check T$, we obtain
\begin{equation} \label{eq:p_s_expo}
    \| p_{i,k+1}-s_{i,k^+} \| \leq \beta_i \| p_{i,k}-s_{i,k^+} \|,
\end{equation}
where $\beta_i = C_ie^{-\lambda_i \check T}$ as in~\eqref{eq:beta_i_def}.

Similar to the shifted variables $\tilde p_{i,k}$ and $\tilde r_{i,k^+}$, we define $\tilde s_{i,k^+} = s_{i,k^+} - d_i$ and the stacked vector $\tilde s_{k^+} = [\tilde s_{1,k^+}^\top, \dots, \tilde s_{n,k^+}^\top]^\top$.
Using~\eqref{eq:s_i_def} and the jump update $\tilde r_{k^+}=(H_k\otimes I_2)\tilde p_k$, we obtain
\begin{align}
    \tilde s_{k^+} &= \big( (I_n-\Sigma_k)\otimes I_2 \big) \tilde p_k + \big(\Sigma_k\otimes I_2 \big) \tilde r_{k^+} \nonumber\\
    &=\big( (I_n-\Sigma_k)\otimes I_2 \big) \tilde p_k + \big((\Sigma_k H_k)\otimes I_2\big) \tilde p_k \nonumber \\
    &=\big( H_k^\sigma\otimes I_2 \big) \tilde p_k,
    \label{eq:s_tilde}
\end{align}
where $\Sigma_k=\diag(\sigma_{1,k},\dots,\sigma_{n,k})$ and 
\begin{equation} \label{eq:H_sigma_def}
    H_k^\sigma = I_n-\Sigma_k + \Sigma_k H_k.
\end{equation}
Note that $H_k^\sigma$ is row-stochastic since $H_k$ is row-stochastic and each row of $H_k^\sigma$ is a convex combination of the corresponding rows of $I_n$ and $H_k$.
The following lemma quantifies how this modification degrades the ergodicity coefficient of $L$-step products.
\begin{lemma} \label{lemma:mixing_mod}
    Under Assumption~\ref{as:mixing} and the definition of $H_k^\sigma$ in~\eqref{eq:H_sigma_def}, it holds for all $k\in\N_0$ that
    \begin{equation}
        \tau_1(H_{k+L-1}^\sigma\cdots H_k^\sigma) \leq 1-\check\sigma^L\mu,
    \end{equation}
    where $\check\sigma=\inf_{i,k} \sigma_{i,k}$ and $L,\mu$ as given in Assumption~\ref{as:mixing}.
\end{lemma}
\begin{proof}
    We begin with the entry-wise inequalities
    \begin{equation}
        H_k^\sigma \geq \Sigma_k H_k \geq \check\sigma H_k.
    \end{equation}
    Consequently,
    \begin{equation}
        H_{k+L-1}^\sigma \cdots H_k^\sigma \geq \check\sigma^L H_{k+L-1}\cdots H_k.
    \end{equation}
    Let $\Pi_k = H_{k+L-1} \cdots H_k$ and $\Pi_k^\sigma = H_{k+L-1}^\sigma \cdots H_k^\sigma$.
    Then,
    \begin{align}
        \min_{i,j} \sum_{s=1}^n &\min \{ [\Pi_k^\sigma]_{is},[\Pi_k^\sigma]_{js} \}  \nonumber\\
        &\geq \check\sigma^L\min_{i,j} \sum_{s=1}^n \min \{ [\Pi_k]_{is},[\Pi_k]_{js} \}.
    \end{align}
    By~\cite[Theorem 3.1]{seneta2006}, the coefficient of ergodicity of any row-stochastic matrix $A$ can equivalently be expressed as
    \begin{equation}
        \tau_1(A) = 1-\min_{i,j} \sum_{s=1}^n \min \{ [A]_{is},[A]_{js} \}.
    \end{equation}
    By that and Assumption~\ref{as:mixing}, we have
    \begin{equation}
        \min_{i,j} \sum_{s=1}^n \min \{ [\Pi_k]_{is},[\Pi_k]_{js} \} \geq \mu.
    \end{equation}
    Combining the above yields
    \begin{equation}
        \tau_1(H_{k+L-1}^\sigma \cdots H_k^\sigma) \leq 1-\check\sigma^L\mu.
    \end{equation}
\end{proof}
Lemma~\ref{lemma:mixing_mod} therefore shows that relaxing Assumption~\ref{as:expo_stable}, and thereby enlarging the required end-of-interval region by allowing $\sigma_{i,k}\leq 1$, reduces the guaranteed disagreement contraction from $\mu$ to $\check\sigma^L\mu$.

\begin{theorem} \label{thm:relaxed}
    Suppose Assumptions~\ref{as:mixing} and~\ref{as:expo_stable_relaxed} are satisfied.
    Then, the jump-flow system defined by~\eqref{eq:jump_r}-\eqref{eq:flow_r}, \eqref{eq:s_i_def}, and \eqref{eq:p_s_expo} achieves the desired formation in the sense of~\eqref{eq:form_cond_convergence}, if for all $i\in\mathcal{N}$ and for all $k\in\N_0$,
    \begin{align}
        \sigma_{i,k}\beta_i &< \frac{1}{2}\big((1+\check\sigma^L\mu)^{\tfrac{1}{L}}-1\big).
    \end{align}
\end{theorem}
\begin{proof}
    The proof of this theorem largely follows the proof of Theorem~\ref{thm:general}.
    
    By~\eqref{eq:s_i_def}, we have $\| p_{i,k}-s_{i,k^+} \|=\sigma_{i,k}\| p_{i,k}-r_{i,k^+} \|$, and therefore the bound in~\eqref{eq:p_s_expo} yields
    \begin{equation} \label{eq:p_s_r_bound}
        \| p_{i,k+1}-s_{i,k^+} \| \leq \beta_i \sigma_{i,k}\| p_{i,k}-r_{i,k^+} \|.
    \end{equation}
    Since $\|p_{i,k}-r_{i,k^+}\| = \|\tilde p_{i,k}-\tilde r_{i,k^+}\|$, applying~\eqref{eq:p_r_Delta_ineq} gives,
    \begin{equation}
        \| p_{i,k+1}-s_{i,k^+} \| \leq \beta_i\sigma_{i,k} \Delta(\tilde p_k).
    \end{equation}

    By defining $v_k=\tilde p_{k+1} - \tilde s_{k^+}$, we can therefore establish the bound
    \begin{align}
        \Delta(v_k) \leq 2\hat\delta \Delta(\tilde p_k),
    \end{align}
    with $\hat\delta = \sup_{i,k} \beta_i\sigma_{i,k}$.
    The steps to obtain this bound are analogous to~\eqref{eq:Delta_w}.
    
    Using the definitions of $v_k$ and $\tilde s_{k^+}$ in~\eqref{eq:s_tilde}, we can write the one-step recursion
    \begin{equation}
        \tilde p_{k+1} = (H_k^\sigma \otimes I_2)\tilde p_k + v_k.
    \end{equation}
    Iterating this over an $L$-step window, applying the seminorm $\Delta$, and using Lemma~\ref{lemma:mixing_mod} gives
    \begin{equation}
        \Delta(\tilde p_{k+L}) \leq (1-\check\sigma^L\mu)\Delta(\tilde p_k) + \big((1+2\hat\delta)^L-1\big)\Delta(\tilde p_k).
    \end{equation}
    Convergence to the formation in the sense of~\eqref{eq:form_cond_convergence} is therefore ensured whenever
    \begin{equation}
        (1-\check\sigma^L\mu) + \big((1+2\hat\delta)^L-1\big) < 1
    \end{equation}
    or equivalently
    \begin{equation}
        \hat\delta < \frac{1}{2}\big((1+\check\sigma^L\mu)^{\tfrac{1}{L}}-1\big).
    \end{equation}
    Since by definition $\forall i,\,\beta_i\sigma_{i,k}\leq \hat\delta$, it suffices to require
    \begin{equation}
        \sigma_{i,k}\beta_i < \frac{1}{2}\big((1+\check\sigma^L\mu)^{\tfrac{1}{L}}-1\big),
    \end{equation}
    which is the stated sufficient condition.
    This concludes the proof.
\end{proof}

Theorem~\ref{thm:relaxed} permits a geometric interpretation.
The bound in~\eqref{eq:p_s_r_bound} states that, at the end of each flow interval, the position $p_{i,k+1}$ lies inside a ball centered at $s_{i,k^+}$ with radius
\begin{equation} \label{eq:radius}
    R_{i,k} = \beta_i\sigma_{i,k} \| p_{i,k}-r_{i,k^+} \|.
\end{equation}
Thus, the product $\beta_i\sigma_{i,k}$ acts as a radius multiplier that quantifies how far the agent may deviate from the segment point $s_{i,k^+}$.

The parameter $\sigma_{i,k}$ can be interpreted as a certificate of the tracking control on flow interval $k$.
It certifies that the closed-loop motion in each flow interval admits exponential convergence to a point on the line segment between the position at the last communication instant $p_{i,k}$ and the received reference $r_{i,k^+}$.
In this view, agents continue to track $r_{i,k^+}$, however, when an agent cannot track $r_{i,k^+}$ sufficiently fast, this analysis may certify convergence to a closer point on the segment, corresponding to a smaller $\sigma_{i,k}$.

On the other hand, $\beta_i$ represents the contraction factor to the segment point itself.
Since the condition in Theorem~\ref{thm:relaxed} has to hold for $\check\sigma$ as well, we can define
\begin{equation}
    \beta^* = \frac{1}{2\check\sigma}\big((1+\check\sigma^L\mu)^{\tfrac{1}{L}}-1\big).
\end{equation}
Consequently, by Theorem~\ref{thm:relaxed}, convergence is assured if $\sigma_{i,k}\beta_i<\check\sigma\beta^*$ is satisfied for all $i\in\mathcal{N}$ and $k\in\N_0$.
Substituting this into the definition of the radius in~\eqref{eq:radius} yields
\begin{equation}
    R_{i,k} < \check\sigma\beta^* \| p_{i,k}-r_{i,k^+} \|,
\end{equation}
which is constant in each flow interval for any choice of $\sigma_{i,k}$.

This yields the following picture.
Instead of requiring every agent to end each interval inside a ball centered at $r_{i,k^+}$, it suffices that they end up within the distance of $R_{i,k}$ to an arbitrary point in between $(1-\check\sigma)p_{i,k}+\check\sigma r_{i,k^+}$ and $r_{i,k^+}$, which can be geometrically interpreted as a rounded tube around that segment with radius $R_{i,k}$.
An example of this consequence is illustrated in Figure~\ref{fig:set_comp} by the yellow shaded region, where to simplify the illustration we take the radius of the dashed circle to be $\beta^*$.

Importantly, the ability to certify $\sigma_{i,k}<1$ depends on the shape of the tracking transient.
Controllers that keep the agent's position near the line segment joining $p_{i,k}$ and $r_{i,k^+}$, make it easier to certify the bound on $\sigma_{i,k}\beta_i$.
In this sense, beyond enlarging $\check T$, the tracking controller itself provides a second lever for guaranteeing consensus and thereby convergence to the desired formation.

Consequently, Theorem~\ref{thm:relaxed} generalizes not only Theorem~\ref{thm:general}, but also the findings in~\cite{molinari2019}, where the tracking control satisfied $\forall k,\forall i,\, \sigma_{i,k}\beta_i=0$ for any choice of $\sigma_{i,k}\in(0,1]$, i.e., the agents moved along the line connecting $p_{i,k}$ and $r_{i,k^+}$, which immediately satisfies the condition in Theorem~\ref{thm:relaxed}.

%% file: src/06_simulation.tex
\section{Simulation Results} \label{sec:simulation}
In this section, we present numerical simulations of swarms implementing different controller parameters to underline the results in Theorems~\ref{thm:general} and~\ref{thm:relaxed}.
To this end, we consider a swarm of $n=6$ agents with unicycle dynamics given by
\begin{equation}
\begin{cases}
    \dot x_i(t) = v_i(t)\cos\theta_i(t) \\
    \dot y_i(t) = v_i(t)\sin\theta_i(t) \\
    \dot \theta_i(t) = \omega_i(t), 
\end{cases}
\end{equation}
with position $p_i(t) = [x_i(t),y_i(t)]^\top$ and inputs $v_i(t)$ and $\omega_i(t)$.
In order to let the agents converge toward their reference positions during flow intervals $t\in(t_k,t_{k+1}]$, we calculate the inputs as 
\begin{align}
    v_i(t) &= K_{v,i} \rho_i(t) \cos\alpha_i(t) \\
    \omega_i(t) &= -K_{\omega,i} \alpha_i(t),
\end{align}
where $K_{v,i}, K_{\omega,i} \in\R_{>0}$ represent the longitudinal and angular control gains, respectively, while $\rho_i(t)$, $\alpha_i(t)$ are given by 
\begin{align}
    \rho_i(t) &= \| p_i(t) - r_{i,k^+} \| \\
    \alpha_i(t) &= \textrm{wrap} \big(\theta_i(t) - \atantwo(r_{i,k^+}^y-y_i(t), r_{i,k^+}^x-x_i(t)) \big),
\end{align}
respectively, with $r_{i,k^+}^x$ and $r_{i,k^+}^y$ denoting the $x$- and $y$-components of $r_{i,k^+}$.

The gain $K_{v,i}$ determines the longitudinal speed of agent $i$, while $K_{\omega,i}$ influences how quickly its orientation is adjusted toward the reference position. 
Consequently, the relation between both gains shapes the resulting trajectories.
In particular, smaller angular gains lead to curved or spiraling trajectories, whereas larger angular gains generally result in a more direct approach toward the reference.
It is straightforward to show that, for a fixed reference position, the closed-loop dynamics are exponentially stable toward the reference position, thus satisfying Assumption~\ref{as:expo_stable}.

In the following, we run three Monte Carlo simulations with 50 independent realizations implementing the reference updates described in Section~\ref{sec:control} together with the agent dynamics introduced above.
We generate a random sequence of strongly connected matrices $H_k$ with a ring-topology satisfying Assumption~\ref{as:mixing} and consider a fixed flow interval length, i.e., $\check T=\hat T = T$.

The initial positions $p_i(0)$ and headings $\theta_i(0)$ are drawn uniformly in $[-5,5]\times[-5,5]$ and $[0,2\pi)$, respectively.
The desired formation is a regular hexagon of radius $5$, i.e., $d_i = 5\left[\cos\frac{2\pi(i-1)}{6},\sin \frac{2\pi(i-1)}{6}\right]^\top$.

For the first two experiments we fix $T=\SI{0.1}{\second}$ and draw the longitudinal controller gains independently according to $\forall i,\, K_{v,i}\sim \mathcal{U}(1.5,2.5)$.
In the first experiment, we choose $\forall i,\, K_{\omega,i}\sim\mathcal{U}(2,3)$, while in the second we set $\forall i,\, K_{\omega,i}\sim\mathcal{U}(1,2)$.
The third experiment uses the same gain distributions as the second, but increases the interval length to $T=\SI{1}{\second}$.
The control gains are calculated once and remain fixed throughout the experiments.

To quantify the convergence to the desired formation, we compute at each communication instant the mean squared error
\begin{equation}
    \mathrm{MSE}(\tilde p_{1,k},\dots,\tilde p_{n,k}) = \frac{1}{n} \sum_{i=1}^n \| \tilde p_{i,k}-\bar p_{k} \|^2,
\end{equation}
where $\bar p_k$ is the arithmetic mean of $\tilde p_{1,k},\dots,\tilde p_{n,k}$.
The formation is achieved when this error converges to zero.

Figure~\ref{fig:sim_n6_mse} shows the averaged formation error over time for the three Monte Carlo experiments.
\begin{figure}
    \centering
    \includegraphics{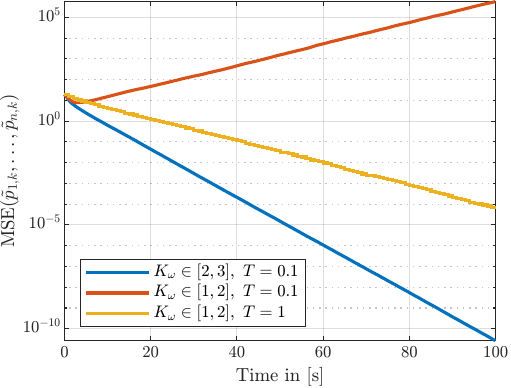}
    \caption{Formation error over time for varying choices of $K_\omega$ and $T$, averaged over 50 independent Monte Carlo realizations.}
    \label{fig:sim_n6_mse}
\end{figure}
In the first experiment ($K_\omega\in[2,3]$, $T=\SI{0.1}{\second}$), the agents reach consensus and the swarms converge to the desired formation in all realizations.
In contrast, in the second simulation ($K_\omega\in[1,2]$, $T=\SI{0.1}{\second}$) the formation error diverges.
This behavior is consistent with the geometric interpretation of Theorem~\ref{thm:relaxed}. 
Due to smaller angular gains, the agents adjust their orientation more slowly, causing trajectories to deviate further from the line segment between the previous position and the updated reference, which in turn makes it harder to satisfy the inequality in Theorem~\ref{thm:relaxed}.

However, convergence is recovered in the third experiment by increasing the time in between communication instants ($K_\omega\in[1,2]$, $T=\SI{1}{\second}$).
This is consistent with Theorem~\ref{thm:general}, which guarantees convergence for sufficiently long flow intervals.

Lastly, we analyze the communication efficiency of the proposed protocol for a single Monte Carlo realization.
For $T=\SI{0.1}{\second}$, a $\SI{100}{\second}$ simulation contains $1000$ communication instants.
Since all agents transmit simultaneously and three scalars must be communicated at each update, the communication protocol requires three orthogonal channels per update, which amounts to $3000$ orthogonal transmissions in total.

For comparison, under a node-to-node protocol each active communication link would require two orthogonal channels.
Since the considered ring topology contains six inter-agent communication links, this amounts to $12000$ orthogonal transmissions over the simulation horizon.
Therefore, the employed protocol provides a substantial reduction in the number of required orthogonal channels. 

%% file: src/07_conclusion.tex
\section{Conclusion} \label{sec:conclusion}
In this paper we proposed a consensus-based formation control scheme for heterogeneous agents.
Instead of imposing explicit assumptions on the agent dynamics, we only require exponential stabilizability of the position with respect to a constant reference.
We derived sufficient conditions that guarantee convergence to the desired formation.
In particular, convergence can be ensured by choosing the time in between communication updates sufficiently large.
Moreover, convergence can also be established through an appropriate choice of tracking control laws that yield favorable transient geometry.
In this way we extend the results in~\cite{borzone2018} and~\cite{molinari2019}.
Numerical simulations illustrate the theoretical findings.

A key element of the approach is the OTA broadcast protocol.
The protocol exploits the superposition property of the wireless channel to compute normalized convex combinations from simultaneous transmissions.
This yields a substantial reduction in the required number of orthogonal channels when compared to traditional interference-avoiding node-to-node communication schemes.

In future work, we plan to incorporate additional adverse communication effects, including additive noise and half-duplex communication.
In addition, we will study the impact of imperfect position measurements and other sensing limitations as well as explicit agent constraints, such as collision avoidance, on the convergence guarantees.

%% file: lib.bib
@book{seneta2006,
  title = {Non-Negative Matrices and Markov Chains},
  author = {Seneta, E.},
  year = 2006,
  series = {Springer Series in Statistics},
  edition = {2nd},
  publisher = {Springer},
  address = {New York, NY},
  doi = {10.1007/0-387-32792-4},
}

@inproceedings{molinari2018,
    author={Molinari, Fabio and Stanczak, Slawomir and Raisch, Jörg},
    booktitle={2018 European Control Conference (ECC)}, 
    title={Exploiting the Superposition Property of Wireless Communication For Average Consensus Problems in Multi-Agent Systems}, 
    year={2018},
    volume={},
    number={},
    pages={1766-1772},
    doi={10.23919/ECC.2018.8550555}
}

@inproceedings{epp2024,
  title = {Exploiting Over-The-Air Consensus for Collision Avoidance and Formation Control in Multi-Agent Systems},
  booktitle = {2024 IEEE 63rd Conference on Decision and Control (CDC)},
  author = {Epp, Michael and Molinari, Fabio and Raisch, J{\"o}rg},
  year = 2024,
  month = dec,
  pages = {5417--5423},
  issn = {2576-2370},
  doi = {10.1109/CDC56724.2024.10886580},
}

@inproceedings{molinari2019,
    author={Molinari, Fabio and Raisch, Jörg},
    booktitle={2019 IEEE 58th Conference on Decision and Control (CDC)}, 
    title={Efficient Consensus-based Formation Control With Discrete-Time Broadcast Updates}, 
    year={2019},
    volume={},
    number={},
    pages={4172-4177},
    doi={10.1109/CDC40024.2019.9029346}
}

@inproceedings{borzone2018,
  title = {Hybrid Framework for Consensus in Fleets of Non-Holonomic Robots},
  booktitle = {2018 Annual American Control Conference (ACC)},
  author = {Borzone, T. and Morarescu, I.-C. and Jungers, M. and Boc, M. and Janneteau, C.},
  year = 2018,
  month = jun,
  pages = {4299--4304},
  publisher = {IEEE},
  address = {Milwaukee, WI},
  doi = {10.23919/ACC.2018.8430918},
}

@article{ren2007,
  title = {Information Consensus in Multivehicle Cooperative Control},
  author = {Ren, Wei and Beard, Randal W. and Atkins, Ella M.},
  year = 2007,
  month = apr,
  journal = {IEEE Control Systems Magazine},
  volume = {27},
  number = {2},
  pages = {71--82},
  issn = {1941-000X},
  doi = {10.1109/MCS.2007.338264},
}

@article{gulzar2018,
  title = {Multi-Agent Cooperative Control Consensus: A Comparative Review},
  shorttitle = {Multi-Agent Cooperative Control Consensus},
  author = {Gulzar, Muhammad Majid and Rizvi, Syed Tahir Hussain and Javed, Muhammad Yaqoob and Munir, Umer and Asif, Haleema},
  year = 2018,
  month = feb,
  journal = {Electronics},
  volume = {7},
  number = {2},
  issn = {2079-9292},
  doi = {10.3390/electronics7020022},
  copyright = {http://creativecommons.org/licenses/by/3.0/},
  langid = {english},
}

@article{olfati-saber2007,
  title = {Consensus and Cooperation in Networked Multi-Agent Systems},
  author = {{Olfati-Saber}, Reza and Fax, J. Alex and Murray, Richard M.},
  year = 2007,
  month = jan,
  journal = {Proceedings of the IEEE},
  volume = {95},
  number = {1},
  pages = {215--233},
  issn = {0018-9219},
  doi = {10.1109/JPROC.2006.887293},
  urldate = {2025-07-30},
}

@article{falconi2015,
  title = {Edge-Weighted Consensus-Based Formation Control Strategy with Collision Avoidance},
  author = {Falconi, Riccardo and Sabattini, Lorenzo and Secchi, Cristian and Fantuzzi, Cesare and Melchiorri, Claudio},
  year = 2015,
  month = feb,
  journal = {Robotica},
  volume = {33},
  number = {2},
  pages = {332--347},
  issn = {0263-5747, 1469-8668},
  doi = {10.1017/S0263574714000368},
  urldate = {2026-01-31},
}

@article{ren2007b,
  title = {Consensus Strategies for Cooperative Control of Vehicle Formations},
  author = {Ren, W.},
  year = 2007,
  month = mar,
  journal = {IET Control Theory \& Applications},
  volume = {1},
  number = {2},
  pages = {505--512},
  issn = {1751-8644, 1751-8652},
  doi = {10.1049/iet-cta:20050401},
  urldate = {2026-02-05},
  langid = {english}
}

@article{xie2014,
  title = {Position Centroid Rendezvous and Centroid Formation of Multiple Unicycle Agents},
  author = {Xie, Wenjing and Ma, Baoli},
  year = 2014,
  journal = {IET Control Theory \& Applications},
  volume = {8},
  number = {17},
  pages = {2055--2061},
  issn = {1751-8652},
  doi = {10.1049/iet-cta.2013.0940},
}

@inproceedings{dimarogonas2008,
  title = {On the Stability of Distance-Based Formation Control},
  booktitle = {2008 47th {{IEEE Conference}} on {{Decision}} and {{Control}}},
  author = {Dimarogonas, Dimos V. and Johansson, Karl H.},
  year = 2008,
  month = dec,
  pages = {1200--1205},
  issn = {0191-2216},
  doi = {10.1109/CDC.2008.4739215},
}

@article{tahbaz-salehi2008,
  title = {A Necessary and Sufficient Condition for Consensus over Random Networks},
  author = {Tahbaz-Salehi, Alireza and Jadbabaie, Ali},
  year = 2008,
  month = apr,
  journal = {IEEE Transactions on Automatic Control},
  volume = {53},
  number = {3},
  pages = {791--795},
  issn = {0018-9286, 1558-2523, 2334-3303},
  doi = {10.1109/TAC.2008.917743},
}

@article{chatterjee1977,
  title = {Towards Consensus: Some Convergence Theorems on Repeated Averaging},
  author = {Chatterjee, S. and Seneta, E.},
  year = 1977,
  month = mar,
  journal = {Journal of Applied Probability},
  volume = {14},
  number = {1},
  pages = {89--97},
  issn = {0021-9002, 1475-6072},
  doi = {10.2307/3213262},
}

@article{wang2024,
  title = {Over-the-Air Computation for {{6G}}: Foundations, Technologies, and Applications},
  shorttitle = {Over-the-Air Computation for {{6G}}},
  author = {Wang, Zhibin and Zhao, Yapeng and Zhou, Yong and Shi, Yuanming and Jiang, Chunxiao and Letaief, Khaled B.},
  year = 2024,
  month = jul,
  journal = {IEEE Internet of Things Journal},
  volume = {11},
  number = {14},
  pages = {24634--24658},
  issn = {2327-4662},
  doi = {10.1109/JIOT.2024.3405448},
}

@article{sordalen1995,
  title = {Exponential Stabilization of Nonholonomic Chained Systems},
  author = {Sordalen, O.J. and Egeland, O.},
  year = 1995,
  month = jan,
  journal = {IEEE Transactions on Automatic Control},
  volume = {40},
  number = {1},
  pages = {35--49},
  issn = {1558-2523},
  doi = {10.1109/9.362901},
}

@article{dewit1992,
  title = {Exponential Stabilization of Mobile Robots with Nonholonomic Constraints},
  author = {{de Wit}, C.C. and Sordalen, O.J.},
  year = 1992,
  month = nov,
  journal = {IEEE Transactions on Automatic Control},
  volume = {37},
  number = {11},
  pages = {1791--1797},
  issn = {1558-2523},
  doi = {10.1109/9.173153},
}

@article{goldenbaum2013,
  title = {Harnessing {{Interference}} for {{Analog Function Computation}} in {{Wireless Sensor Networks}}},
  author = {Goldenbaum, Mario and Boche, Holger and Sta{\'n}czak, S{\l}awomir},
  year = 2013,
  month = oct,
  journal = {IEEE Transactions on Signal Processing},
  volume = {61},
  number = {20},
  pages = {4893--4906},
  issn = {1941-0476},
  doi = {10.1109/TSP.2013.2272921},
}
